\newif\ifarxiv
\arxivtrue
\ifarxiv%
	\documentclass[11pt]{article}
	\usepackage[margin=3.75cm]{geometry}
	\usepackage{authblk}
	\usepackage{amsthm}
\else%
	\documentclass{llncs}
\fi

\usepackage{appendix-tools}
\ifarxiv
	\appendixfalse
\fi

\usepackage{xspace}
\usepackage[english]{babel}
\usepackage[utf8]{inputenc}
\usepackage{indentfirst}
\usepackage{amsmath}
\usepackage{amsfonts}
\usepackage{amssymb}
\usepackage{xcolor,graphicx}
\usepackage{tabularx}
\usepackage{mathtools}
\usepackage{thmtools}
\usepackage{thm-restate}
\usepackage{subfig}
\usepackage{colorpallette}
\usepackage{dsfont}
\usepackage{tikz}
\usepackage{paralist}
\usepackage{algorithm}
\usepackage[noend]{algpseudocode}
\usepackage[toc,page]{appendix}
\usepackage{mathtools}

\usepackage{graphdrawing}
\usepackage{pictikzgraph}
\usetikzlibrary{decorations.pathreplacing, decorations.markings, arrows}
\usepackage{myprobdescr}
\usepackage{mydefs}
\usepackage{problemdefs}
\usepackage[binary-units]{siunitx}
\usepackage{float}
\usepackage{booktabs}
\usepackage{textcomp}
\ifarxiv
\usepackage[sort]{natbib}
\renewcommand{\cite}[1]{\citep{#1}}
\fi
\usepackage{breakurl}
\usepackage{hyperref}
\hypersetup{colorlinks=true, linkcolor=darkblue, citecolor=darkblue, breaklinks=true}
\usepackage{breakcites}
\usepackage{nameref,thmtools}
\usepackage[sort&compress,nameinlink,noabbrev,capitalize]{cleveref}
\usepackage{todonotes}
\usepackage{paralist}
\usepackage{tablefootnote}
\usepackage{etoolbox}
\usepackage{hacks}

\makeatletter
\g@addto@macro{\UrlBreaks}{\UrlOrds}
\makeatother

\newcommand{\SetLabelname}{Set Label}
\newcommand{\SatLabelname}{Satisfy Label}
\newcommand{\SetLabel}{\hyperref[rrule:funnels:arc deletion set label]{\SetLabelname}\xspace}
\newcommand{\SatLabel}{\hyperref[rrule:funnels:arc deletion satisfy label]{\SatLabelname}\xspace}
\newcommand{\allglab}{\ForkF/\MergeF/-labeling}
\newcommand{\funlab}{funnel labeling\xspace}
\newcommand{\optlab}{optimal labeling\xspace}
\newcommand{\ArcDelApprox}{\text{\upshape\Algname{ArcDeletionSet}}\xspace}%

\newcommand{\pkL}{\textsc{$k$-Linkage}}
\newcommand{\pADDFl}{\textsc{Arc-Deletion Distance to a Funnel}}
\newcommand{\pADDF}{\textsc{ADDF}}

\tikzstyle{vecArrow} = [thick, decoration={markings,mark=at position
	1 with {\arrow[semithick]{open triangle 60}}},
double distance=1.4pt, shorten >= 5.5pt,
preaction = {decorate},
postaction = {draw,line width=1.4pt, white,shorten >= 4.5pt}]

\ifarxiv
	\theoremstyle{plain}
	\newtheorem{theorem}{Theorem}
	\newtheorem{observation}{Observation}
	\newtheorem{definition}{Definition}
	\newtheorem{proposition}{Proposition}
	\newtheorem{lemma}{Lemma}
	
	\theoremstyle{definition}
	\newtheorem{case}{Case}
	\newtheorem{rrule}{Reduction Rule}
	\newtheorem{rrule*}{Reduction Rule}
	\newtheorem{brule}{Branching Rule}
\else
	\spnewtheorem{observation}{Observation}{\bfseries}{\itshape}
	\spnewtheorem{rrule}{Reduction Rule}{\bfseries}{\itshape}
	\spnewtheorem{rrule*}{Reduction Rule}{\bfseries}{\itshape}
	\spnewtheorem{brule}{Branching Rule}{\bfseries}{\itshape}
\fi

\crefname{rrule}{Reduction Rule}{Reduction Rules}
\crefname{case}{Case}{Cases}
\crefname{thm}{Theorem}{Theorems}
\crefname{align}{Equation}{Equations}
\crefname{ineq}{Inequality}{Inequalities}
\crefname{brule}{Branching Rule}{Branching Rules}
\crefname{observation}{Observation}{Observations}
 
\title{Efficient Algorithms for Measuring the Funnel-likeness of DAGs}

\ifarxiv
	\author[1]{Marcelo~Garlet~Millani\thanks{Partially supported by DFG project ``FPTinP'' NI 369/16-1.}}

	\affil[1]{Institut f\"ur Softwaretechnik und Theoretische Informatik,
 TU Berlin, Germany,
 
 \texttt{\{m.garletmillani. h.molter, rolf.niedermeier\}@tu-berlin.de}}

	\author[1]{Hendrik~Molter}

	\author[1]{Rolf~Niedermeier}

	\author[2]{Manuel~Sorge\thanks{Supported by the People Programme (Marie Curie Actions) of the European Union's Seventh Framework Programme (FP7/2007-2013) under REA grant agreement number {631163.11} and Israel Science Foundation (grant no. 551145/14).}}

	\affil[2]{Dept.\ Industrial Engineering and Management, Ben-Gurion University of the Negev, Beer Sheva, Israel,
	
\texttt{sorge@post.bgu.ac.il}}

\else
	\author{Marcelo~Garlet~Millani\inst{1}\thanks{Partially supported by DFG project ``FPTinP'' NI 369/16-1.} \and Hendrik~Molter\inst{1} \and Rolf~Niedermeier\inst{1}
  \and Manuel~Sorge\inst{1,2}\thanks{Supported by the People Programme (Marie Curie Actions) of the European Union's Seventh Framework Programme (FP7/2007-2013) under REA grant agreement number {631163.11} and
Israel Science Foundation (grant no. 551145/14).}}

	\institute{Institut f\"ur Softwaretechnik und Theoretische Informatik,
  TU Berlin, Germany \email{\{m.garletmillani, h.molter, rolf.niedermeier\}@tu-berlin.de} \and Department of Industrial Engineering and Management, Ben-Gurion University of the Negev, Beer Sheva, Israel,\\  
  \email{sorge@post.bgu.ac.il}}
\fi

\begin{document}

\maketitle 

\begin{abstract}
  Funnels are a new natural subclass of DAGs.
  Intuitively, a DAG is a funnel if every source-sink path can be uniquely identified by one of its arcs.
  Funnels are an analog to trees for directed graphs that is more restrictive than DAGs but more expressive than in-/out-trees.
  Computational problems such as finding vertex-disjoint paths or tracking the origin of memes remain NP-hard on DAGs while on funnels they become solvable in polynomial time.
  Our main focus is the algorithmic complexity of finding out how funnel-like a given DAG is.
  To this end, we study the NP-hard problem of computing the arc-deletion distance to a funnel of a given DAG.
  We develop efficient exact and approximation algorithms for the problem and test them on synthetic random graphs and real-world graphs.
\end{abstract}



\section{Introduction}
\label{sec:intro}
Directed acyclic graphs (DAGs) are finite directed graphs (digraphs)
without directed cycles and appear in many applications, including the representation of
precedence constraints in scheduling, data processing networks, causal
structures, or inference in proofs. From a more graph-theoretic point of view, DAGs can
be seen as a directed analog of trees; however, their combinatorial structure is much richer.
Thus a number of directed graph problems remain NP\nobreakdash-hard even when
restricted to DAGs. This motivates the study of subclasses of DAGs.
We study \emph{funnels} which are DAGs where each source-sink path has at least one private arc, that
is, no other source-sink path contains this arc. 
In independent work, Lehmann~\cite{Lehmann17} studied essentially the same graph class.

Funnels are both of combinatorial and graph-theoretic as well as of
practical interest: First, funnels are a natural compromise between
DAGs and trees as, similarly to in- or out-trees, the private-arc
property guarantees that the overall number of source-sink paths is
upper-bounded linearly by its number of arcs, yet multiple paths
connecting two vertices are possible.
%
Second, in \cref{sec:funnels} we show that funnels, in a divide \& conquer spirit, allow for a vertex
partition into a set of \emph{forking} vertices with indegree one and
possibly large outdegree and a set of \emph{merging} vertices with
outdegree one and possibly large indegree. This partitioning helps in designing our algorithms. Third, in terms of applications,
due to the simpler structure of funnels, problems
such as \textsc{DAG Partitioning}~\cite{Bevern2016,Leskovec2009}
or \textsc{Vertex Disjoint Paths}, (also known as
\pkL)~\cite{Gutin-Digraphs-2008,Fortune:1978} become tractable on funnels while they are NP-hard
on DAGs. Lehmann~\cite{Lehmann17} showed that a variation of the problem \textsc{Network
Inhibition}, which is NP-hard on DAGs, can be solved in polynomial time on funnels.
Altogether, we feel that funnels are one of so far few natural subclasses of DAGs.

The focus of this paper is on investigating the complexity of turning a given DAG into a funnel by a minimum number of arc deletions.
The motivation for this is twofold.
First, due to the noisy nature of real-world data, we expect that graphs from practice are not pure funnels, even though they may adhere to some form of funnel-like structure.
To test this hypothesis we need
efficient algorithms to determine funnel-likeness.
Second, as mentioned above, natural computational problems become tractable on
funnels (e.g., \pkL\ \cite{Mil17}). Thus it is promising to try and develop
fixed-parameter algorithms for such NP-hard DAG problems with respect to
distance parameters to funnels. This approach is known as exploiting
the ``distance from triviality''~\cite{Cai03,GHN04,Nie10}. A
natural way to measure the distance of a given DAG~$D$ to a
funnel is the \emph{arc-deletion distance to a funnel}, the minimum
number of arcs that need to be deleted from $D$ to obtain a funnel.
The problem of computing this distance parallels the well-studied NP-hard \FeedbackArcSet\ problem where the task is to turn a given digraph into a DAG by a minimum number of arc deletions. Even \FeedbackArcSet\ on tournaments is NP-hard and it received considerable interest over the last years~\cite{ailon2007hardness,bessy2011kernels,charbit2007minimum,kenyon2007rank}.

Formally, we study the \pADDFl~(\pADDF) problem, where, given a DAG~$D$, we want to find its arc-deletion distance~$d$ to a funnel.
%
We show that \pADDF\ is NP-hard and that it admits a linear-time factor-two approximation algorithm and a fixed-parameter algorithm with linear running time for constant~$d$.%
\footnote{There is also a simple
  $\Bo(5^d\cdot\Abs{V}\cdot\Abs{A})$-time algorithm for general digraphs \cite{Mil17}.}
  In experiments we demonstrate that our algorithms are useful in practice.

  \ifappendix
  Due to the lack of space, proofs of results marked with (\appsymb{}) are deferred to an Appendix.\fi

\section{Funnels: Definition and Properties}
\label{sec:funnels}
\appendixsection{sec:funnels}
\newcommand{\Fa}{\ensuremath{\mathcal{F}}\xspace}

In this section we formally define funnels.
We provide several equivalent characterizations, summarized in \cref{thm:funnel:characterization}, and analyze some basic properties of funnels. We use standard terminology from graph theory.

	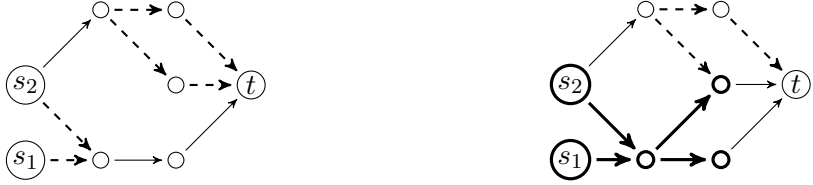
\begin{figure}[t]
	\centering
	\begin{Subfigure}{\Half}
		\begin{tikzpicture}[rotate=90]
\node[pictikz-node, draw=pictikz-black] (path4485-9) at (1.0, 0.0) {$t$};
\node[pictikz-node, draw=pictikz-black] (path4485-7) at (0.0, 1.0) {};
\node[pictikz-node, draw=pictikz-black] (path4485-35) at (1.0, 1.0) {};
\node[pictikz-node, draw=pictikz-black] (path4485-62) at (2.0, 1.0) {};
\node[pictikz-node, draw=pictikz-black] (path4485-5) at (2.0, 2.0) {};
\node[pictikz-node, draw=pictikz-black] (path4485-6) at (0.0, 2.0) {};
\node[pictikz-node, draw=pictikz-black] (path4485) at (0.0, 3.0) {$s_1$};
\node[pictikz-node, draw=pictikz-black] (path4485-3) at (1.0, 3.0) {$s_2$};
\draw[draw=pictikz-black, pictikz-dashed, pictikz-edgeto] (path4485) edge (path4485-6);
\draw[draw=pictikz-black, pictikz-edgeto] (path4485-6) edge (path4485-7);
\draw[draw=pictikz-black, pictikz-dashed, pictikz-edgeto] (path4485-3) edge (path4485-6);
\draw[draw=pictikz-black, pictikz-edgeto] (path4485-3) edge (path4485-5);
\draw[draw=pictikz-black, pictikz-dashed, pictikz-edgeto] (path4485-5) edge (path4485-35);
\draw[draw=pictikz-black, pictikz-dashed, pictikz-edgeto] (path4485-5) edge (path4485-62);
\draw[draw=pictikz-black, pictikz-dashed, pictikz-edgeto] (path4485-62) edge (path4485-9);
\draw[draw=pictikz-black, pictikz-dashed, pictikz-edgeto] (path4485-35) edge (path4485-9);
\draw[draw=pictikz-black, pictikz-edgeto] (path4485-7) edge (path4485-9);
		\end{tikzpicture}
	\end{Subfigure}
	\qquad
	\begin{Subfigure}{\Half}
		\begin{tikzpicture}[rotate=90]
\node[pictikz-node, draw=pictikz-black] (path4485-9) at (1.0, 0.0) {$t$};
\node[pictikz-node, draw=pictikz-black, pictikz-thick] (path4485-7) at (0.0, 1.0) {};
\node[pictikz-node, draw=pictikz-black, pictikz-thick] (path4485-35) at (1.0, 1.0) {};
\node[pictikz-node, draw=pictikz-black] (path4485-62) at (2.0, 1.0) {};
\node[pictikz-node, draw=pictikz-black] (path4485-5) at (2.0, 2.0) {};
\node[pictikz-node, draw=pictikz-black, pictikz-thick] (path4485-6) at (0.0, 2.0) {};
\node[pictikz-node, draw=pictikz-black, pictikz-thick] (path4485) at (0.0, 3.0) {$s_1$};
\node[pictikz-node, draw=pictikz-black, pictikz-thick] (path4485-3) at (1.0, 3.0) {$s_2$};
\draw[draw=pictikz-black, pictikz-thick, pictikz-edgeto] (path4485) edge (path4485-6);
\draw[draw=pictikz-black, pictikz-thick, pictikz-edgeto] (path4485-6) edge (path4485-7);
\draw[draw=pictikz-black, pictikz-thick, pictikz-edgeto] (path4485-3) edge (path4485-6);
\draw[draw=pictikz-black, pictikz-edgeto] (path4485-3) edge (path4485-5);
\draw[draw=pictikz-black, pictikz-dashed, pictikz-edgeto] (path4485-5) edge (path4485-35);
\draw[draw=pictikz-black, pictikz-dashed, pictikz-edgeto] (path4485-5) edge (path4485-62);
\draw[draw=pictikz-black, pictikz-dashed, pictikz-edgeto] (path4485-62) edge (path4485-9);
\draw[draw=pictikz-black, pictikz-edgeto] (path4485-35) edge (path4485-9);
\draw[draw=pictikz-black, pictikz-edgeto] (path4485-7) edge (path4485-9);
\draw[draw=pictikz-black, pictikz-thick, pictikz-edgeto] (path4485-6) edge (path4485-35);
		\end{tikzpicture}
	\end{Subfigure}
	\caption{Example of a funnel (left) and a DAG which is not a funnel (right).
	Private arcs are marked as dashed lines.
	The DAG on the right is not a funnel because all arcs in an \Patht{s_1}{t} are shared.
	Removing one arc from it turns it into a funnel.
	A forbidden subgraph for funnels is marked in bold.}
	\label{fig:funnel:example}
	\end{figure}
	To define funnels as a proper subclass of DAGs, we limit the number of paths that may exist between two vertices (which can be exponential in DAGs but is one in trees).
	Requiring every path between two vertices to be unique would possibly be too restrictive, and in the case of a single source such DAGs would simply be so-called out-trees.
	Instead, we require each path going from a source to a sink to be uniquely identified by one of its \emph{private} arcs.
	We say that an arc is private if there is only one source-sink path which goes through that arc.
	An example of a funnel can be seen in \Cref{fig:funnel:example}.
	\begin{definition}[Funnel]
		\label{def:funnel}
		A DAG~$D$ is a \emph{funnel} if every source-sink path has at least one private arc.
		\label{def:funnel:funnel = private edge}
	\end{definition}
	From this definition it is clear that the number of source-sink paths in a funnel is linearly upper-bounded in its number of arcs.
	
	Different characterizations of funnels reveal certain interesting properties which these digraphs have, and are used in subsequent proofs and algorithms.
	We summarize these characterizations in the theorem below.
	In the following, $\OutC{v}$ denotes the set of vertices that can be reached from $v$ in a given DAG,~$\OutN{v}$ denotes the set of neighbors of~$v$ and~$\Out{v}$ denotes $v$'s outdegree;
	$\InC{v}, \InN{v}$ and $\In{v}$ are defined analogously.
	
	\ifarxiv\pagebreak\fi
	\begin{theorem}
		\label{thm:funnel:characterization}
 		Let~$D$ be a DAG. The following statements are equivalent:
		\begin{compactenum}
			\item $D$ is a funnel.
			\label{chr:funnel}
			
			\item For each vertex~$v \in V: \In{v} > 1 \Rightarrow \forall u \in \OutC{v}:  \Out{u} \leq 1$.
			\label{chr:degree}
			
			\item No subgraph of~$D$ is contained in~$\Fa = \ManyZ{D}{\infty}$, where
			\begin{compactitem}
			\item $D_k = (V_k, A_k)$,
			\item $V_k = \{u_1, u_2, v_0, w_1,w_2\} \cup \Many{v}{k}$, and 
			\item $A_k = \{(u_1,v_0), (u_2, v_0), (v_k, w_1), (v_k, w_2)\} \cup \{(v_i, v_{i+1})\}_{i=1}^{k-1}$.
			\end{compactitem} 
			\label{chr:forbidden subgraph}
			\item $D$ does not contain~$D_0$ or~$D_1$ (defined above) as a topological minor.\footnote{A graph $H$ is called a \emph{topological minor} of a graph $G$ if a subgraph of $G$ can be obtained from $H$ by subdividing edges (that is, replacing arcs by directed paths).
                          }
			\label{chr:forbidden minor}
		\end{compactenum}
	\end{theorem}
	\appendixproof{thm:funnel:characterization}{
	\begin{figure}
	\centering
		\begin{tikzpicture}[rotate=90, xscale=1.5, yscale=2.5]
		\input{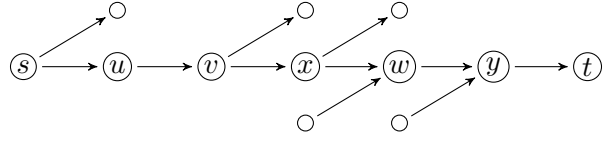}
		\end{tikzpicture}
	\caption{Illustration of the vertices used in the proof of statment (\ref{chr:degree}) \Cref{thm:funnel:characterization}.
	We argue that the arc~$(x,w)$ is private.}
	\label{fig:funnel:prop degree}
	\end{figure}
	\begin{proof}We first prove that (\ref{chr:funnel}) $\Leftrightarrow$ (\ref{chr:degree}),
		that is, we show that a DAG $D = (V,A)$ is a funnel if and only if 
		\begin{align}
			&\label[ineq]{item:neighbor degree} \forall v \in V: \In{v} > 1 \Rightarrow \forall u \in \OutC{v}:  \Out{u} \leq 1.
		\end{align}
		The idea is to identify the private arcs and to argue that each path must contain at least one of those arcs.
		Refer to \cref{fig:funnel:prop degree} while reading the proof.
		We start by showing that a DAG satisfying \eqref{item:neighbor degree} is a funnel.

		Let $D = (V,A)$ be a DAG which satisfies \eqref{item:neighbor degree}, let $s \in V$ be a source and $t \in V$ be a sink such that some \Patht{s}{t} exists, and let~$v$ be the first vertex in $\OutC{s}$ with $\Out{v} > 1$.
		If no such vertex $v$ exists, then there is only one \Patht{s}{t} in~$D$ and all outgoing arcs from~$s$ are private.
		Otherwise, due to \eqref{item:neighbor degree} we know that~$\forall u \in \InC{v} \setminus \{s\} : \In{u} = 1$.
		This means that there is exactly one \Patht{s}{v}.
		Let~$P$ be some \Patht{s}{t} that goes through~$v$ and let~$w$ be the first vertex in this path with $\In{w} > 1$.
		If no such~$w$ exists, then there is only one \Patht{v}{t} and all arcs after $v$ are private, as required.
		Otherwise, we consider a vertex~$x$ such that the arc~$(x,w)$ is in~$P$.
		We know~$\In{x} = 1$, which implies that there is only one \Patht{v}{x}.
		Since the \Patht{v}{x} as well as the \Patht{s}{v} are unique and $\forall y \in \OutC{w} : \Out{y} = 1$, the arc~$(x,w)$ is private for the \Patht{s}{t} that contains it.
		Thus,~$D$ is a funnel.

		We next show that every funnel satisfies \eqref{item:neighbor degree}.
		We do this by contraposition, showing that every arc of some \Patht{s}{t} is present in at least one other source-sink path if \eqref{item:neighbor degree} does not hold.

		Let~$D = (V,A)$ be DAG where \eqref{item:neighbor degree} is not true.
		This means that there is some vertex~$u \in V$ with~$\In{u} > 1$ and that there is some other vertex $w \in \OutC{u}$ with $\Out{w}>1$.
		Let~$w$ be the first such vertex.
		Then there are at least $\In{u}$~many paths from some source to $u$, and $\Out{w}$~many from $w$ to some sink.
		Since there is at least one \Patht{u}{w} (possibly without arcs), this implies that every arc in the induced subgraph~$\InCc{u}$ is shared by $\Out{w}$~many paths, every arc in~$\OutCc{w}$ is shared by $\In{u}$~many paths, and all arcs in a \Patht{u}{w} are shared by~$\In{u} \cdot \Out{w}$~many paths.
		Hence, all arcs in a source-sink path which goes through~$u$ and~$w$ are shared, implying that $D$ is not a funnel.

		Next, we prove that $(\ref{chr:degree}) \Leftrightarrow (\ref{chr:forbidden subgraph})$ by contraposition.
		That is, we show that~$\exists C \subseteq D : C \in \Fa$ if and only if~$D$ does not satisfy \eqref{item:neighbor degree}.

		Let $C \subseteq D$ be a subgraph of~$D$ such that $C \in \Fa$.
		By definition of~$C$ it contains some vertex~$v_0$ with~$\In[C]{v_0} = 2$ and another vertex~$v_k \in \OutC[C]{v_0}$ with~$\Out[C]{v_k} = 2$.
		This implies~$\In[D]{v_0} > 1$ and~$\Out[D]{v_k} > 1$, violating~\eqref{item:neighbor degree}.

		Now assume~$D$ is does not satisfy \eqref{item:neighbor degree}.
		That is, there is some vertex~$v$ with $\In[D]{v} > 1$ and another vertex~$u \in \OutC[D]{v}$ with~$\Out[D]{u} > 1$.
		Let~$u_1, u_2 \in \InN[D]{v}$ and~$w_1, w_2 \in \OutN[D]{u}$ be four distinct vertices.
		Let~$v,v_1,v_2,$ $\dots ,v_{k-1},u$ be a \Patht{v}{u}.
		We set~$v_0 \Set v$ and~$v_k \Set u$, obtaining the forbidden subgraph~$D_k$ if~$u \neq v$, and~$D_0$ otherwise.
		Hence,~$D$ contains a subgraph from~$\Fa$.

	Finally, we show that $(\ref{chr:forbidden subgraph}) \Leftrightarrow (\ref{chr:forbidden minor})$.
	It is enough to show that any~$D_i \in \Fa$ can be obtained by subdividing~$D_0$ or~$D_1$ multiple times, and that any subdivision of~$D_0$ and~$D_1$ contains some digraph of~$\Fa$ as a subgraph.

	We first show that we can generate \Fa by subdividing~$D_0$ and~$D_1$.
	Let~$D_i \in \Fa$.
	If~$i \leq 1$, then~$D_i$ obviously contains~itself as a topological minor.
	If~$i > 1$, then by subdividing the arc~$(v_0,v_1)$ from~$D_1$ a total of~$i-1$ times, we obtain~$D_i$.
	Hence all digraphs in \Fa can be generated by~$D_0$ and~$D_1$ through subdivisions.

	Now we show that any subdivision of~$D_0$ and~$D_1$ contains some digraph from~\Fa.
	Since subdividing arcs does not change the degrees of the affected vertices, the degree of $v_0$ remains the same.
	Hence, any subdivision of $D_0$ contains $D_0 \in \Fa$ as a subgraph.

	For any subdivision~$D'_1$ of~$D_1$ we know that~$\In[D'_1]{v_0} = 2$, $\Out[D'_1]{v_1} = 2$ and~$v_1 \in \OutC[D'_1]{v_0}$.
	If we subdivide incoming arcs of $v_0$ or outgoing arcs of $v_1$, the resulting DAG will contain $D_1$ as a subgraph.
	If we subdivide $k$ times the arc $(v_0, v_1)$, we obtain $D_{k+1} \in \Fa$ as a subgraph.
	
	We showed that $(\ref{chr:funnel}) \Leftrightarrow (\ref{chr:degree}) \Leftrightarrow (\ref{chr:forbidden subgraph}) \Leftrightarrow (\ref{chr:forbidden minor})$, thus proving that all four statements are equivalent.
	\end{proof}}

	\Cref{def:funnel} does not give us a very efficient way of checking whether a given DAG is a funnel or not.
	A simple algorithm which counts how many paths go through each arc would take~$\Bo(\Abs{A}^2)$ time.
	Using the characterization in \cref{thm:funnel:characterization}(\ref{chr:degree}) we can follow some topological ordering of the vertices of a DAG and check in linear time whether it is a funnel.
	
	\looseness=-1 The \emph{degree characterization} in \cref{thm:funnel:characterization}(\ref{chr:degree}) provides some additional insight about the structure of a funnel.
	We can see that a funnel can be partitioned into two induced subgraphs: One is an out-forest and the other is an in-forest.
	Note that this partition is not necessarily unique.
	For use below, a \emph{\allglab} for given a DAG with vertex set~$V$ is a function~$L : V \rightarrow \{\Fork, \Merge\}$ which gives a \emph{label} to each vertex. An \allglab\ for a funnel is called \emph{\funlab} if the 
	vertices in the out-forest of the funnel are assigned the label \Fork and vertices in the in-forest are assigned the label \Merge.
	 The following holds.
	\begin{observation}\label{obs:funnel:labels}
	Let~$D = (V, A)$ be a funnel and~$L$ be a \funlab for~$D$.
	Then there is no $(v,u) \in A$ with $L(v) = \Merge$ and $L(u) = \Fork$.
	\end{observation}
	
	\noindent With a simple counting argument it is also possible to give an upper bound on the number of arcs in a funnel.
	This bound is sharp.

	\begin{observation}
	\label{prop:funnel:number of arcs}
	Let~$D = (V,A)$ be a funnel. Then~$\Abs{A} \leq \Abs{V}^2/4 + \Abs{V} - 2$.
	\end{observation}
	\appendixproof{prop:funnel:number of arcs}{
	\begin{proof}
	Let~$L$ be a \funlab\ for~$D$.
	Let~$V = X \uplus Y$ where~$\forall v \in X: L(v) = \Fork$ and $\forall v \in Y: L(v) = \Merge$.
	Clearly, the vertices in~$X$ form an out-forest, while those in~$Y$ form an in-forest.
	This gives us at most~$\Abs{X}-1$ arcs between vertices in~$X$, and at most~$\Abs{Y} - 1$ arcs between vertices in~$Y$.
	Furthermore, there are at most~$\Abs{X}\cdot\Abs{Y}$ arcs from vertices in~$X$ to vertices in~$Y$ and we know from the construction that there are no arcs from~$Y$ to~$X$.
	Hence,~$\Abs{A} \leq \Abs{X}\cdot\Abs{Y} + \Abs{X} + \Abs{Y} - 2$.
	This value is maximized when~$\Abs{X} = \Abs{Y} = \Abs{V}/2$, which gives us the bound $\Abs{A} \leq \Abs{V}^2/4 + \Abs{V} - 2$.
	\end{proof}}
		Considering that a DAG has at most $\Abs{V}(\Abs{V} - 1)/2$ arcs, \cref{prop:funnel:number of arcs} implies that a funnel can have roughly half as many arcs as a DAG.
		This means that funnels are not necessarily sparse (unlike forests).


While the degree characterization is useful for algorithms, the characterizations by forbidden subgraphs and minors (\cref{thm:funnel:characterization}(\ref{chr:forbidden subgraph} and \ref{chr:forbidden minor})) help us to understand the local structure of a funnel and of graphs that are not funnels.
These characterizations also imply that being a funnel is a \emph{hereditary} graph property, that is, deleting vertices does not destroy the funnel property.

\section{Computing the Arc-Deletion Distance to a Funnel}
\label{sec:algs}
\appendixsection{sec:algs}

\newcommand{\Apf}{b}
\newcommand{\Opf}{a}

In this section we show \pADDF\ is NP-hard, and present a linear-time factor-2 approximation algorithm and an exact fixed-parameter algorithm.
Our algorithms also compute the set of arcs to be deleted.
We remark that the corresponding vertex-deletion distance minimization problem is also NP-hard and that it can be solved in $\Bo(6^d\Abs{V} \cdot \Abs{A})$ time, where $d$ is the number of vertices to delete \cite{Mil17}.
The following result can be shown by a reduction from \textsc{3-SAT}.
\begin{theorem}
\label{thm:funnels:arc deletion:np-hard}
\pADDF\ is NP-hard.
\end{theorem}%
\appendixproof{thm:funnels:arc deletion:np-hard}{%
\begin{proof}%
\begin{figure}[t]
	\centering
	\begin{tikzpicture}
		\input{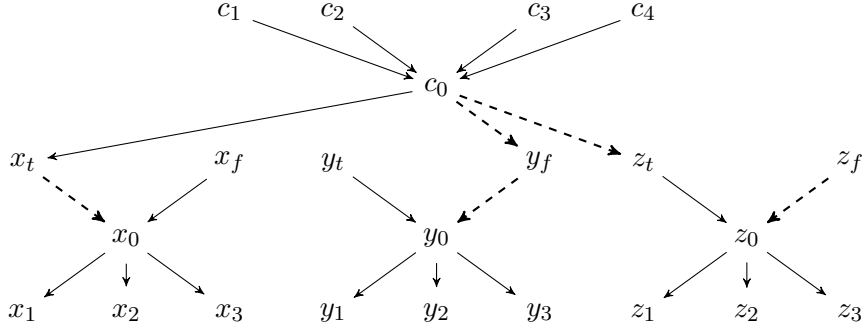}
	\end{tikzpicture}
	\caption{Example of the reduction for the formula $(x \lor \lnot y \lor z)$. Dashed arcs correspond to a solution for \pADDF\ on the reduced instance.}\label{fig:arcdelnphard}
\end{figure}%
We present a reduction from \ThreeSAT. Recall that in \ThreeSAT\ we are asked to decide the satisfiability of given a Boolean formula~$\phi$ in conjunctive normal form where every clause has exactly three distinct literals. Given a \ThreeSAT\ formula $\phi$ with $n$ variables and $m$ clauses, we create a DAG~$D$ consisting of the following variable gadgets and clause gadgets. \Cref{fig:arcdelnphard} illustrates the construction. For each variable $x$ we create the following \emph{variable gadget} introducing the vertex set $V_x$ and edge set~$A_x$:
\begin{compactitem}
\item $V_x=\{x_0, x_t, x_f, x_1, x_2, x_3\}$,
\item $A_x=\{(x_t, x_0), (x_f, x_0)\}\cup \{(x_0, x_i)\mid 1\le i\le 3\}$.
\end{compactitem}
We call $x_0$ the \emph{center} of the variable gadget for $x$. For each clause $c$, we create the following \emph{clause gadget}, introducing the vertex set~$V_c$ and edge set~$A_c$: 
\begin{compactitem}
\item $V_c=\{c_0, c_1, \ldots, c_4\}$,
\item $A_c=\{(c_i, c_0)\mid 1\le i\le 4\}$.
\end{compactitem}
We call $c_0$ the \emph{center} of the clause gadget for $c$. Furthermore, if variable $x$ appears non-negated in clause $c$, then we add the arc~$(c_0, x_t)$, and if variable $x$ appears negated in clause $c$, then we add the arc~$(c_0, x_f)$.
This completes the construction. It is easy to see that the DAG~$D$ can be constructed in polynomial time. We claim that $D$ has an arc-deletion distance to funnel of $k=2m+n$ if and only if $\phi$ is satisfiable.

$(\Leftarrow)$: Assume $\phi$ has a satisfying assignment. Then we construct an arc-deletion set of size $k=2m+n$ as follows: If a variable $x$ is set to true, we delete the arc $(x_t, x_0)$, otherwise we delete the arc $(x_f, x_0)$. For each clause $c$ we delete two of the three outgoing arcs of $c_0$, where we choose the remaining arc to be one that points to a literal that causes the clause to be satisfied by the assignment. This arc deletion set clearly has the correct size, it remains to show that it destroys all forbidden subgraphs of funnels in the constructed DAG. Note that after the arcs are deleted, there are only two types of vertices with indegree greater than one: The centers of clause gadgets and potentially vertices $x_t$ or $x_f$ from variable gadgets. The only vertices with outdegree greater than one remaining are the centers of variable gadgets. Because the outgoing arcs of clause gadgets point to literals that cause the clause to be satisfied, we have that all paths from clause gadget centers to vertex gadget centers are destroyed. By the same argument, there are no paths between vertices $x_t$ or $x_f$ from variable gadgets that have indegree greater than one and centers of variable gadgets. Hence, there is no path from a vertex with indegree greater than one to a vertex with outdegree greater than one.

$(\Rightarrow)$: First, note that all variable and clause gadgets are pair-wise arc-disjoint. It is easy to check that for each variable gadget at least one arc needs to be deleted and for each clause gadget at least two arcs need to be deleted. Since the number of arc deletions has to be at most $k=2m+n$, the arc deletion set contains exactly one arc from each variable gadget and exactly two arcs from each clause gadget. This implies that for clause gadgets, the two of the outgoings arcs of the center need to be deleted and for variable gadgets, one of the incoming arcs of the center needs to be deleted. We claim that the arcs deleted from the variable gadgets induce a satisfying assignment in a straightforward manner: if the arc~$(x_t, x_0)$ is deleted, set variable $x$ to true, otherwise to false. Take any clause $c$ of $\phi$, one of the outgoing arcs from the center of the clause gadget of $c$ remains, and this arc has to point to a vertex with outdegree zero, otherwise there is a path from a center of a clause gadget to a center of a variable gadget and hence a forbidden subgraph. This means that clause $c$ is satisfied. This completes the proof.
\end{proof}}%

\subsubsection*{A Factor-2 Approximation Algorithm.}
\label{subsec:funnel:arc deletion:approximation}
\newcommand{\Label}[1][]{\Funname{label}_{#1}}%
\newcommand{\ApLabel}{\ensuremath{L_a}}
\newcommand{\OpLabel}{\ensuremath{L^*}}
\newcommand{\InLabel}{\ensuremath{L}}
We now give a linear-time factor\nobreakdash-2 approximation algorithm for \pADDF.
We mention in passing that on tournament DAGs the algorithm always finds an optimal solution and on real-world DAGs, the approximation factor is typically close to one (see \cref{sec:experiments}). 
The approximation algorithm works in three phases and makes extensive use of \allglab s (defined in \Cref{sec:funnels}).
First, we greedily compute an \allglab\ which we call~$\ApLabel$ for the input graph (assigning each vertex~$v$ a \Fork or a \Merge label).
The labeling will be a funnel labeling of the output funnel indicating for each vertex whether it can have indegree or outdegree greater than one.
To construct~$\ApLabel$, we try to minimize the number of arcs to be removed when only considering~$v$.
This strategy guarantees that, if the approximation algorithm assigns the wrong label to~$v$, in the optimal solution many arcs incident to $v$ need to be removed.
This allows us to derive the approximation factor.
Formally, we assign a label to a vertex~$v$ using the following rule.
\[
	\ApLabel(v) \Set
	\begin{cases}
		\Fork,  & \text{if }\Out[D]{v} > \In[D]{v},\\
		\Fork,  & \text{if }\Out[D]{v} = \In[D]{v} \land \phantom{}\\
                & \phantom{} \exists u \in \InN{v} : \ApLabel(u) = \Fork,\\
		\Merge, & \text{otherwise.}
	\end{cases}
\]
Since we can assign a label whenever we know the labels of all incoming neighbors, the label of each vertex can be computed, in linear time, by following a topological ordering of the DAG.

In the second phase, after assigning labels to all vertices, we \emph{satisfy} the labels by removing arcs.
That is, for each \Fork vertex~$v$, we choose an arbitrary inneighbor $u$ with $\InLabel(u) = \Fork$ (if it exists) and remove all arcs incoming to~$v$ from vertices other than~$u$. Similarly, for each \Merge vertex~$v$ we choose an arbitrary outneighbor $u$ with $\InLabel(u) = \Merge$ (if it exists) and remove all arcs outgoing from~$v$ to vertices other than~$u$.
See \cref{alg:funnel:arc deletion:approximation} for the pseudocode of the second phase. For use below we call the second-phase algorithm \ArcDelApprox.
\begin{algorithm}[t]
\begin{algorithmic}[1]
	\Function{ArcDeletionSet}{DAG~$D = (V,A)$, $\InLabel : V \rightarrow \{\Fork, \Merge\}$}
        \State $B \gets \emptyset$
        \ForAll{$v \in V$}
        \If{$\InLabel(v) = \Merge$}
        \State Choose an arbitrary $u \in \OutN{v}$ with $\InLabel(u) = \Merge$ (if it exists)
        \State $B \gets B \cup \{(v, w) \mid w \neq u \wedge w \in \OutN{v}\}$\label{alg:funnel:arc-deletion approximation:w1}
        \label{alg:funnels:arc-deletion approximation:merge}
        \ElsIf{$\InLabel(v) = \Fork$}
        \State Choose an arbitrary $u \in \InN{v}$ with $\InLabel(u) = \Fork$ (if it exists)
        \State $B \gets B \cup \{(w, v) \mid w \neq u \wedge w \in \InN{v}\}$
        \EndIf
        \EndFor
        \State \Return $B$
	\EndFunction
\end{algorithmic}
\caption{Satisfying an \allglab.}
\label{alg:funnel:arc deletion:approximation}
\end{algorithm}%

\looseness=-1 In the third phase, we \emph{greedily relabel} vertices, that is, we iterate over each vertex $v$ (in an arbitrary order), changing $v$'s label if the change immediately leads to an improvement in the solution size.
To check if there is an improvement, we only need to consider the incident arcs of $v$ and the labels of its endpoints.
This completes the description of our approximation algorithm.

To argue about optimal solutions and for use in a search-tree algorithm below, we now show that if the input \allglab\ $L$ corresponds to an optimal solution, then \ArcDelApprox\ outputs an optimal arc set: Say that an \allglab~$L$ of a DAG~$D$ is \emph{optimal} if it is a \funlab for some funnel~$D - A'$, $A' \subseteq A$, such that $A'$ has minimum size among all arc sets whose deletion makes $D$ a funnel.
\begin{proposition}
\label{prop:funnels:arc deletion approx=opt}
Let~$D = (V,A)$ be a DAG, let~$A' \subseteq A$ be a minimum arc set such that~$D' = D - A'$ is a funnel, and let $\OpLabel$ be an \optlab\ for~$D'$. Then $\Abs{\text{\ArcDelApprox}(D, \OpLabel)} = \Abs{A'}$.
\end{proposition}%
\appendixproof{prop:funnels:arc deletion approx=opt}{%
\begin{proof}	
	Let $(v, u) \in A$.	
  We distinguish the possible cases of the labeling of $u$ and $v$. First, we treat two simple cases in which we can argue that $A'$ and \ArcDelApprox either both contain $(v, u)$ or both do not contain~$(v, u)$.

  The first case is when~$\OpLabel(v) = \Merge$ and~$\OpLabel(u) = \Fork$.
  Then $(v, u)$ has to be both in~$A'$ as well as in the solution given by \ArcDelApprox, which we call from now on~$B$.  
	It is clearly in~$B$ since it was added to the solution on \Cref{alg:funnels:arc-deletion approximation:merge} of \cref{alg:funnel:arc deletion:approximation}.
	Due to \cref{obs:funnel:labels}, we know that $(v,u)$ is also in $A'$.

The second case is when~$\OpLabel(v) = \Fork$ and~$\OpLabel(u) = \Merge$.
In this case, clearly, removing $(v, u)$ will not destroy any forbidden subgraph, since~$\forall w \in \InC[D'] {v} : \OpLabel(w) = \Fork$.
Since \ArcDelApprox does not remove the arc, it is neither present in~$B$ nor in~$A'$.

For the remaining cases we cannot guarantee that exactly the same decision was taken with respect to $(v, u)$.
We instead argue about the total number of arcs removed between vertices with the same label.
From \cref{thm:funnel:characterization}(\ref{chr:degree}) we know that \Fork vertices form an induced outforest in $D$, while \Merge vertices form an induced inforest.
The number of arcs in an in- or outforest is given by the number of vertices minus the number of roots (i.e.\ sources or sinks).
All incoming arcs of a \Fork vertex $v$ are removed by \ArcDelApprox only if $v$ has no inneighbors labeled with \Fork.
Hence, $v$ is a source in $D - B$ if and only if it is a source in $D - A'$.
This implies that the number of arcs in the outforest composed of \Fork vertices is the same in $D - B$ as in $D - A'$.
An analogous argument holds for the inforest induced by \Merge vertices.
Hence, the total number of arcs between equally labeled vertices is the same in $B$ and $A'$.
Since these were all cases and in all of them \ArcDelApprox deletes as many arcs as the optimal solution, we conclude that~$\Abs{B} = \Abs{A'}$.
%
%
\end{proof}
}

We now give a guarantee of the approximation factor.
\toappendix{Example for the tightness of the approximation factor of two}%
{app:approx factor example}%
{%
A DAG where the approximation algorithms removes twice as many arcs as an optimal solution is given in \cref{fig:approx factor example}.
\begin{figure}[t]
	\centering
	\begin{tikzpicture}
		\input{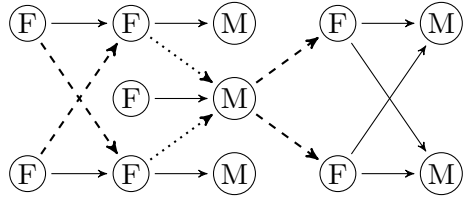}
	\end{tikzpicture}
	\caption{Example of the execution of \Algname{ArcDeletionSet}.
Vertices with an F received the label \Fork, and those with an M received the label \Merge.
The approximation algorithm returns the four dashed arcs, while there is an optimal solution (dotted arcs) of size two.
Note that changing any single label will not improve the approximate solution.}
	\label{fig:approx factor example}
\end{figure}
}
\ifappendix%
Due to space constraints we provide only a proof sketch. The technical details are in \cref{proof:thm:arc deletion:approx factor 2} and an example for the tightness of the approximation factor is in \cref{app:approx factor example}.
\fi

\begin{theorem}
	\label{thm:arc deletion:approx factor 2}
	There is a linear-time factor-two approximation for \pADDF.
\end{theorem}
\ifappendix%
\begin{proof}[Sketch]
	Let $A'$ be a minimum arc set such that $D - A'$ is a funnel and $B = \ArcDelApprox(D, \ApLabel)$, where $\ApLabel$ is computed by the above described procedure.
Let $\OpLabel$ be an optimal \allglab\ for the input DAG~$D = (V, A)$.
We define two functions $\Apf : V \rightarrow \mathcal{P}(B)$ and $\Opf : V \rightarrow \mathcal{P}(A')$ such that $\biguplus_{v \in V}\Apf(v) = B$ and $\biguplus_{v \in V} \Opf(v) = A'$, where $\biguplus$ is a disjoint union and $\mathcal{P}(X)$ denotes the family of all subsets of a set~$X$.
Our goal is to assign each arc in $A'$ or $B$ to one of its endpoints via $a$ or $b$, respectively, such that $\Abs{b(v)} \leq 2\Abs{a(v)}$ for every $v \in V$.
We say that a vertex $v$ has \emph{type} $T(v) = \ForkF/\MergeF/$ if $\ApLabel(v) = \Fork$ and $\OpLabel(v) = \Merge$.
The types \ForkF/\ForkF/, \MergeF/\MergeF/ and \MergeF/\ForkF/ are defined analogously.
A vertex $v$ is \emph{correctly labeled} if $\ApLabel(v) = \OpLabel(v)$.

We define $a$ and $b$ in such a way that $\Abs{b(v)} = \Abs{a(v)}$ if $v$ is correctly labeled.
To this end, we only assign a removed arc to a correctly labeled vertex $v$ if both endpoints are correctly labeled.
For an incorrectly labeled vertex, we assign the arcs which are potentially removed by \ArcDelApprox when considering $v$, together with those of correctly labeled vertices.
We additionally need to define $a$ and $b$ in such a way that no arc is assigned to both endpoints.

By construction, it is easy to show that $\Abs{\Apf(v)} = \Abs{\Opf(v)}$ if $v$ is correctly labeled.
We now consider an incorrectly labeled vertex $v$.

If $\In{v} = 1 = \Out{v}$, then the approximation removes at most one of the incident arcs of $v$.
If both are removed, then the algorithm changes the label in the third phase, which implies that $v$ was correctly labeled either before or after the change.
As only one arc of $v$ is removed, we can treat $(u,v)$ and $(v,w)$ as a single arc $(u,w)$ and assign it to the same vertex that $(u,w)$ would be assigned to if it would be removed.

For the remaining cases, we use a counting argument based on the amount of neighbors of $v$ with each type.
We additionally need an exchange argument, that is, whenever we have an arc $(v,u) \in A'$ where $T(v) = \ForkF/\MergeF/$ and $T(u) = \MergeF/\ForkF/$, some arc in $\Apf(u)$ needs to be assigned to $\Apf(v)$ instead.
The exchange is possible because such arcs are always in $A'$ but never in $B$, meaning that the approximation has an ``advantage'' over the optimal solution with respect to these arcs.

Because the functions $\Opf$ and $\Apf$ partition $A'$ and $B$, respectively, we obtain that $\Abs{\biguplus_{v \in V}\Apf(v)} = \Abs{B} \leq 2\Abs{A'} = 2\Abs{\biguplus_{v \in V} \Opf(v)}$.
\end{proof}%
\fi%
\appendixproof{thm:arc deletion:approx factor 2}{
\begin{proof}
After computing $B = \ArcDelApprox(D, \ApLabel)$, the approximation algorithm iterates over $D - B$, flipping labels whenever the flip leads to an improvement in the solution.
This implies that, if we remove all incoming arcs of a vertex $v$ with $\ApLabel(v) = \Merge$, then we set the label of $v$ to $\Fork$ instead.
Analogously, we flip the label of $v$ if all of its outgoing arcs have been removed and $\ApLabel(v) = \Fork$.

Let $B = \ArcDelApprox(D, \ApLabel)$ and let $A'$ be a minimum arc set such that $D - A'$ is a funnel.
Let $\OpLabel$ be an optimal \allglab\ for the input DAG~$D = (V, A)$ such that $A' = \ArcDelApprox(D, \OpLabel)$.
We define two functions $\Apf : V \rightarrow \mathcal{P}(B)$ and $\Opf : V \rightarrow \mathcal{P}(A')$ such that $\biguplus_{v \in V}\Apf(v) = B$ and $\biguplus_{v \in V} \Opf(v) = A'$, where $\biguplus$ is a disjoint union and $\mathcal{P}(X)$ denotes the family of all subsets of a set~$X$.
Our goal is to assign each arc in $A'$ and $B$ to one of its endpoints via $a$ and $b$, respectively, such that $\Abs{b(v)} \leq 2\Abs{a(v)}$ for every $v \in V$.
We say that a vertex $v$ has type $T(v) = \ForkF/\MergeF/$ if $\ApLabel(v) = \Fork$ and $\OpLabel(v) = \Merge$.
The types \ForkF/\ForkF/, \MergeF/\MergeF/ and \MergeF/\ForkF/ are defined analogously.
A vertex $v$ is correctly labeled if $\ApLabel(v) = \OpLabel(v)$.

We define $a$ and $b$ in such a way that $\Abs{b(v)} = \Abs{a(v)}$ if $v$ is correctly labeled.
To this end, we only assign a removed arc to a correctly labeled vertex $v$ if both endpoints are correctly labeled.
For an incorrectly labeled vertex, we assign the arcs which are potentially removed by \ArcDelApprox when considering $v$, together with those of correctly labeled vertices.
We additionally need to define $a$ and $b$ in such a way that no arc is assigned to both endpoints.
Refer to \cref{fig:approx:arcs partition} for a graphical representation of $a$ and $b$.
\begin{figure}[t]
		\centering
		\begin{tikzpicture}
\node[inner sep=1.0pt,radius=3pt] (path4136-62) at (0.75, 0.0) [align=left]{\MergeF/\MergeF/};
\node[inner sep=1.0pt,radius=3pt] (path4136-35) at (0.0, 0.0) [align=left]{\ForkF/\ForkF/};
\node[inner sep=1.0pt,radius=3pt] (path4136-9) at (2.25, 0.0) [align=left]{\ForkF/\MergeF/};
\node[inner sep=1.0pt,radius=3pt] (path4136-1) at (3.0, 0.0) [align=left]{\MergeF/\ForkF/};
\node[inner sep=1.0pt,radius=3pt] (path4136) at (1.5, 1.25) [align=left]{\ForkF/\MergeF/};
\node[inner sep=1.0pt,radius=3pt] (path4136-3) at (0.0, 2.5) [align=left]{\ForkF/\ForkF/};
\node[inner sep=1.0pt,radius=3pt] (path4136-7) at (2.25, 2.5) [align=left]{\ForkF/\MergeF/};
\node[inner sep=1.0pt,radius=3pt] (path4136-6) at (0.75, 2.5) [align=left]{\MergeF/\MergeF/};
\node[inner sep=1.0pt,radius=3pt] (path4136-5) at (3.0, 2.5) [align=left]{\MergeF/\ForkF/};
\draw[bend right=25, pictikz-dotted, pictikz-edgeto] (path4136) edge (path4136-35);
\draw[draw=pictikz-black, pictikz-dashed, pictikz-edgeto] (path4136-3) edge (path4136);
\draw[bend left=25, pictikz-dashed, pictikz-edgeto] (path4136-6) edge (path4136);
\draw[draw=pictikz-black, pictikz-dashed, pictikz-edgeto] (path4136-7) edge (path4136);
\draw[draw=pictikz-black, pictikz-dashed, pictikz-edgeto] (path4136-5) edge (path4136);
\draw[draw=pictikz-black, pictikz-dashed, pictikz-edgeto] (path4136) edge (path4136-35);
\draw[draw=pictikz-black, pictikz-dotted, pictikz-edgeto] (path4136) edge (path4136-62);
\draw[draw=pictikz-black, pictikz-dotted, pictikz-edgeto] (path4136) edge (path4136-9);
\draw[draw=pictikz-black, pictikz-dotted, pictikz-edgeto] (path4136) edge (path4136-1);
\draw[draw=pictikz-black, pictikz-dotted, pictikz-edgeto] (path4136-6) edge (path4136);
		\end{tikzpicture}
		\quad
		\begin{tikzpicture} 
			\input{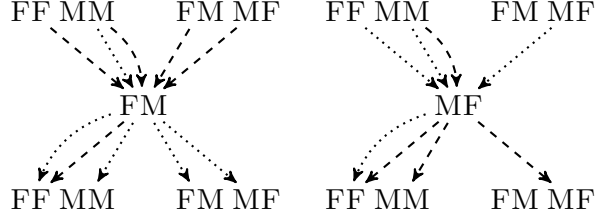}
		\end{tikzpicture}\\[0.5\baselineskip]
		\caption{Graphical representation of \Opf\ and \Apf.
		Vertices are identified with their type.
		Arcs are assigned to the vertex in the middle.		
		Dashed arcs correspond to arcs counted in $\Apf$, while dotted arcs are counted in $\Opf$.}
		\label{fig:approx:arcs partition}
	\end{figure}
\begin{align*}
		\Apf(v) & \Set \begin{cases}
				B \cap \{(u,v) \mid T(u) = \ForkF/\ForkF/\}, & T(v) = \ForkF/\ForkF/,\\
		    B  \cap \{(v,u) \mid T(u) = \ForkF/\ForkF/ \lor T(u) = \MergeF/\MergeF/\},        & T(v) = \MergeF/\MergeF/,\\
		    B \cap ( \{(u,v) \mid u \in \InN{v}\} \cup \{(v,u) \mid T(u) = \ForkF/\ForkF/\} ), & T(v) = \ForkF/\MergeF/,\\
		    B  \cap ( \{ (u,v) \mid T(u) = \MergeF/\MergeF/ \} \cup \{(v,u) \mid T(u) \neq \ForkF/\MergeF/\} ), & T(v) = \MergeF/\ForkF/.%
		    \end{cases}\\
		\Opf(v) & \Set \begin{cases}
		    A' \cap \{(u,v) \mid T(u) = \ForkF/\ForkF/\}, & T(v) = \ForkF/\ForkF/,\\
		    A' \cap \{(v,u) \mid T(u) = \ForkF/\ForkF/ \lor T(u) = \MergeF/\MergeF/\}, & T(v) = \MergeF/\MergeF/,\\
		    A' \cap ( \{(v,u) \mid u \in \OutN{v}\} \cup \{(u,v) \mid T(u) = \MergeF/\MergeF/\} ), & T(v) = \ForkF/\MergeF/,\\
		    A' \cap ( \{(v,u) \mid T(u) = \ForkF/\ForkF/\} \cup \{(u,v) \mid T(u) \neq \ForkF/\MergeF/ \} ), & T(v) = \MergeF/\ForkF/.
		    \end{cases}
	\end{align*}
	
	We now consider each vertex type $t$ and argue that $\Abs{b(v)} \leq 2\Abs{a(v)}$ for every vertex $v$ with $T(v) = t$.
	By construction of $a$ and $b$, this is easy to prove for correctly labeled vertices.
	Further, induced paths of any length behave just like an induced with three vertices, and so we only need to consider the latter case.	
\begin{lemma}
		\label{lemma:funnels:arc deletion approx:L = label}
		If $\ApLabel(v) = \OpLabel(v)$ or $\In{v} = \Out{v} = 1$, then $\Abs{\Apf(v)} = \Abs{\Opf(v)}$.
	\end{lemma}
	{
	\begin{proof}
		When deciding which incident arc of a vertex $v$ is kept, \ArcDelApprox\ makes an arbitrary choice among the valid possibilities.
		However, all choices lead to a solution of the same size.
		Hence, we can assume, without loss of generality, that if \ArcDelApprox\ can keep an arc from a correctly labeled neighbor, then it does so.
		This allows us to assume, for the sake of this analysis, that, if an arc between two correctly labeled vertices is removed by the approximation algorithm, then it is also removed in an optimal solution.
	Formally, we can assume the following for any $v \in V$ which is correctly labeled.
	If $\OpLabel(v) = \Fork$ and there is some correctly labeled $u \in \InN{v}$ with $\OpLabel(u) = \Fork$, then $(w,v) \in A' \cap B$ for all incorrectly labeled $w \in \InN{v}$.
	If $\OpLabel(v) = \Merge$ and there is some correctly labeled $u \in \OutN{v}$ with $\OpLabel(u) = \Merge$, then $(v,w) \in A' \cap B$ for all incorrectly labeled $w \in \OutN{v}$.
	
	We now show for every correctly labeled $v$ that $\Abs{\Apf(v)} = \Abs{\Opf(v)}$.
	We first define variables which count how many neighbors of each type $v$ has.	
	Let $i_{\ForkF/\MergeF/}$ be the number of inneighbors of $v$ with type \ForkF/\MergeF/.
	The variables $i_{\ForkF/\ForkF/}$, $i_{\MergeF/\MergeF/}$ and $i_{\MergeF/\ForkF/}$ are defined analogously, and $o_{\ForkF/\MergeF/}$, $o_{\ForkF/\ForkF/}$, $o_{\MergeF/\MergeF/}$ and $o_{\MergeF/\ForkF/}$ are defined analogously for the outneighbors of $v$.
	
	Let $\OpLabel(v) = \Merge$, then $o_{\ForkF/\ForkF/} \leq \Abs{\Apf(v)} \leq o_{\MergeF/\MergeF/} + o_{\ForkF/\ForkF/}$ and $o_{\ForkF/\ForkF/} \leq \Abs{\Opf(v)} \leq o_{\MergeF/\MergeF/} + o_{\ForkF/\ForkF/}$.
	If $o_{\MergeF/\MergeF/} = 0$, then $\Abs{\Apf(v)} = o_{\ForkF/\ForkF/} = \Abs{\Opf(v)}$.
	Otherwise, due to the initial assumption, $\Abs{\Apf(v)} = o_{\ForkF/\ForkF/} + o_{\MergeF/\MergeF/} - 1 = \Abs{\Opf(v)}$.
	Hence, $\Abs{\Apf(v)} = \Abs{\Opf(v)}$.
	The case where $\OpLabel(v) = \Fork$ follows analogously.
	
	For any path $v_1,v_2\dots v_k$ where all vertices have in- and outdegree one, the approximation assigns the same label to all vertices.
	Furthermore, it removes at most two arcs in such a path.
	The decision of whether to remove an arc or not depends only on the label of the predecessor of $v_1$ and of the successor of $v_k$, and not on the length of the path.
	Hence, we can treat this case by contracting the path into a single vertex $v$.
	Let $u \in \In{v}$ and $w \in \Out{v}$ be the unique neighbors of $v$.
	Note that, by definition, $\ApLabel(v) = \ApLabel(u)$.
	
	If both incident arcs of $v$ are in $B$, we flip the label of $v$ in the greedy relabeling phase.
	Either before or after the flip $v$ is correctly labeled.
	Since flipping the label of $v$ does not worsen the solution, it follows from the previous case that $\Abs{\Apf(v)} = \Abs{\Opf(v)}$.
	
	If only one arc of $v$ was removed, the path $uvw$ behaves as a single arc $(u,w)$, and the removed arc can be assigned to a vertex by considering the types of $u$ and $w$, taking the same decision as if we were assigning the arc $(u,w)$ to a vertex.
	Hence, $\Abs{\Apf(v)} = 0 = \Abs{\Opf(v)}$ in this case.	
	\end{proof}
	}	
	We are now ready to prove an approximation factor of two.

\begin{lemma}
\label{lemma:funnels:arc deletion approximation}
	$\Abs{B} \leq 2 \cdot \Abs{A'}$
\end{lemma}
\begin{proof}
	We first define variables which count the number of neighbors of~$v$ for each type.
	Let $i_{\ForkF/\MergeF/}$ be the number of inneighbors of $v$ with type~\ForkF/\MergeF/.
	The variables $i_{\ForkF/\ForkF/}$, $i_{\MergeF/\MergeF/}$ and $i_{\MergeF/\ForkF/}$ are defined analogously, and $o_{\ForkF/\MergeF/}$, $o_{\ForkF/\ForkF/}$, $o_{\MergeF/\MergeF/}$ and $o_{\MergeF/\ForkF/}$ are defined analogously for the outneighbors of $v$.
	Next, we show for every incorrectly labeled vertex $v$ (with in- or outdegree greater than one) that $\Abs{\Apf(v)} + o_{\MergeF/\ForkF/} \leq 2\Abs{\Opf(v)}$ (if $T(v) = \ForkF/\MergeF/$) and $\Abs{\Apf(v)} - i_{\ForkF/\MergeF/} \leq 2\Abs{\Opf(v)}$ (if $T(v) = \MergeF/\ForkF/$).
	Note that the sum of all $o_{\MergeF/\ForkF/}$ equals the sum of all $i_{\ForkF/\MergeF/}$.
	Hence, we also show that $\Abs{\biguplus_{v \in X}\Apf(v)} \leq 2 \Abs{\biguplus_{v \in X}\Opf(v)}$, where $X$ is the set of all incorrectly labeled vertices.
	\begin{case}{$T(v) = \ForkF/\MergeF/$.}
		By definition, $\Abs{\Apf(v)} \leq c + o_{\ForkF/\ForkF/}$, where $c \leq \In{v}$ is the number of incoming arcs removed from $v$ by $B$.
		Since $\OpLabel(v) = \Merge$, any arc $(v,u)$ with $\OpLabel(u) = \Fork$ must be in $A'$.
		Furthermore, we need to remove at least $\Out{v} - 1$ many arcs from $v$ in order to satisfy its label.
		Hence, $\Abs{\Opf(v)} \geq d + o_{\MergeF/\ForkF/} + o_{\ForkF/\ForkF/} \geq \Out{v} - 1$ for some $0 \leq d \leq \Out{v}$.
		Thus,
		\(
			\Abs{\Apf(v)} + o_{\MergeF/\ForkF/} \leq 2\Abs{\Opf(v)}
			\Leftarrow c + o_{\ForkF/\ForkF/} + o_{\MergeF/\ForkF/} \leq 2(d + o_{\ForkF/\ForkF/} + o_{\MergeF/\ForkF/})
			\Leftarrow c \leq d + \Out{v} - 1.
			\)

		If $\In{v} = \Out{v}$, we know (from the definition of $\ApLabel$) that some inneighbor of $v$ is labeled \Fork\ by $\ApLabel$.
		In this case, $c \leq \In{v} - 1 = \Out{v} - 1$.
		If $\In{v} < \Out{v}$, then $c \leq \In{v} \leq \Out{v} - 1$.
		In both cases, $c \leq d + \Out{v} - 1$ and so $\Abs{\Apf(v)} + o_{\MergeF/\ForkF/} \leq 2\Abs{\Opf(v)}$.
		
	\end{case}
	\begin{case}{$T(v) = \MergeF/\ForkF/$.}
		By definition, $\Abs{\Apf(v)} \leq c + i_{\MergeF/\MergeF/}$, where $c \leq \Out{v}$ is the number of outneighbors of $v$ of type different from~\ForkF/\MergeF/ contained in $B$.
		Since $\OpLabel(v) = \Fork$, all incoming arcs $(u,v)$ with $T(u) = \MergeF/\MergeF/$ must be contained in $A'$.
		If $T(u) = \ForkF/\MergeF/$, then the arc $(v,u)$ is assigned by $\Opf$ to $u$ and not to $v$ (if it is in $A'$).
		Furthermore, we need to remove at least $\In{v} - 1$ arcs, whereas arcs from inneighbors with type \ForkF/\MergeF/ are not counted in $\Opf(v)$.
		Hence, $\Abs{\Opf(v)} \geq d + i_{\MergeF/\MergeF/} \geq \In{v} - 1 -i_{\ForkF/\MergeF/}$, where $0 \leq d \leq \In{v}$.
		It suffices to show that
		\(
			\Abs{\Apf(v)} - i_{\ForkF/\MergeF/} \leq 2\Abs{\Opf(v)}
			\Leftarrow c + i_{\MergeF/\MergeF/} - i_{\ForkF/\MergeF/} \leq 2d + 2i_{\MergeF/\MergeF/}
			\Leftarrow c \leq d + \In{v} - 1.\)
		
		If $\In{v} = \Out{v} \geq 2$, then $i_{\ForkF/\MergeF/} = 0$ since $\ApLabel(v) = \Merge$.
		This implies that $\Abs{\Opf(v)} = d + i_{\MergeF/\MergeF/} \geq \In{v} - 1 = \Out{v} - 1$.
		If $i_{\MergeF/\MergeF/} = 0$, then $\Abs{\Apf(v)} \leq c \leq \Out{v} \leq 2(\Out{v} - 1) = 2(\In{v} - 1) \leq 2\Abs{\Opf(v)}$.
		If $i_{\MergeF/\MergeF/} > 0$, then we argue that at least one incoming arc of $v$ was not removed, since otherwise the label of $v$ would be changed in the greedy relabeling phase of the approximation.
		Hence, $\Abs{\Apf(v)} \leq 2 \Abs{\Opf(v)} \Leftarrow c + i_{\MergeF/\MergeF/} - 1 \leq 2d + 2i_{\MergeF/\MergeF/} \Leftarrow c \leq d + \In{v} = d + \Out{v}$.
		If $\In{v} > \Out{v}$, then $c \leq \Out{v} \leq \In{v} - 1$.		
		In both cases, we have $\Abs{\Apf(v)} - i_{\ForkF/\MergeF/} \leq 2\Abs{\Opf(v)}$.
	\end{case}
	
	Thus, $\Abs{\Apf(v)} \leq 2\Abs{\Opf(v)}$ for any incorrectly labeled vertex~$v$.
	The same holds for correctly labeled vertices by \cref{lemma:funnels:arc deletion approx:L = label}.
	By definition we know that $\biguplus_{v \in V}\Apf(v) = B$ and $\biguplus_{v \in V} \Opf(v) = A'$.
	Hence, $\Abs{\biguplus_{v \in V}\Apf(v)} = \Abs{B} \leq 2\Abs{A'} = 2\Abs{\biguplus_{v \in V} \Opf(v)}$.
\end{proof}

\Cref{alg:funnel:arc deletion:approximation} clearly runs in linear time, as  computing the topological ordering of a DAG can be done in linear time.
The third phase of the algorithm, where labels are changed, can also be executed in linear time by following any ordering of the vertices.
We only change the label of a vertex if this leads to a better solution.
To check if we have a better solution we only need to consider all incident arcs of a vertex and the labels of their endpoints.
Since \cref{lemma:funnels:arc deletion approx:L = label,lemma:funnels:arc deletion approximation} consider all cases for all vertices $v \in V$, we conclude that \pADDF\ can be approximated in linear time within a factor of two.
\end{proof}
}%
\subsubsection*{A Fixed-Parameter Algorithm.}%
\label{subsec:funnel:arc deletion:label}%
 Using the forbidden subgraph characterization
(\cref{thm:funnel:characterization}(\ref{chr:forbidden subgraph})), we can compute a digraph's
arc-deletion distance~$d$ to a funnel in
$\Bo(5^d\cdot(\Abs{V}^2 + \Abs{V}\cdot\Abs{A} ))$~time: After
contracting the arcs on each vertex with in- and outdegree one into a single
arc, it is enough to destroy all subgraphs~$D_0$ or $D_1$ as in
\cref{thm:funnel:characterization}(\ref{chr:forbidden subgraph}). The optimal arc-deletion set to
destroy all these subgraphs can be found by branching into the at most
five possibilities for each subgraph~$D_0$ or~$D_1$.

In this section, we show that, if the input is a DAG, we can solve \pADDF\ in $\Bo(3^d \cdot(\Abs{V} + \Abs{A}))$ time instead; thus, in particular, we have linear running time if $d \in \Bo(1)$.
Moreover, the resulting algorithm has also better running time in practice.
As in the approximation algorithm, we again label the vertices.
\Cref{prop:funnels:arc deletion approx=opt} shows that, after the vertices are correctly labeled with either \Merge or \Fork, solving \pADDF\ can be done in linear time on DAGs.
Hence, the complicated part of the problem lies in finding such a labeling. 

\looseness=-1 In the following, we describe a search-tree algorithm that receives a DAG $D = (V, A)$ and an upper bound~$d \in \Nat$ on the size of the solution as input, and it maintains a partial labeling $L \colon V \to \{\Fork, \Merge\}$ of the vertices and a partial arc-deletion set~$A'$ that will constitute the solution in the end.
Initially, $A' = \emptyset$ and $L(v)$ is undefined for each $v \in V$, denoted by $L(v) = \bot$.
The algorithm exhaustively and alternately applies the data reduction and branching rules described below and aborts if $|A'| > d$.
The rules either determine a label of a vertex (based on preexisting labels and on the degree of the vertex) or put some arcs into the solution~$A'$.
Herein, when we say that an arc is put into the solution, we mean that it is deleted from~$D$ and put into~$A'$.
To show that the algorithm finds a size-$d$ arc deletion set to a funnel if there is one, we ensure that the rules are \emph{correct}, meaning that, if there is a solution of size~$d$ that respects the labeling~$L$ and contains~$A'$ before applying a data reduction rule or branching rule, then there is also such a solution in at least one of the resulting instances.

\Cref{rrule:funnels:arc deletion set label} labels vertices of indegree (outdegree) at most one in a greedy fashion, based on the label of the single predecessor (successor) if it exists.%
\begin{rrule}[\SetLabelname\ifappendix) ({\appref[\appsymb]{proof:rrule:funnels:arc deletion set label}}\fi]
	\label{rrule:funnels:arc deletion set label}
	Let~$v \in V$ be an unlabeled vertex.
	
	Set~$L(v) \Set \Fork$ if at least one of the following is true:
	\begin{inparaenum}[I)]
		\item $\In{v} = 0$;
		\item $\In{v} = 1$ and $\exists u \in \InN{v}: L(u) = \Fork$; 
		\item $\Out{v} > 1$,~$\In{v} = 1$ and~$\forall u \in \OutN{v} : L(u) \neq \bot$.
	\end{inparaenum}

	Set~$L(v) \Set \Merge$ if at least one of the following is true:
	\begin{inparaenum}[I)]
		\item	$\Out{v} = 0$;
		\item	$\Out{v} = 1$ and $\exists u \in \OutN{v}: L(u) = \Merge$;
		\item	$\Out{v} = 1$,~$\In{v} > 1$ and~$\forall u \in \InN{v} : L(u) \neq \bot$.
	\end{inparaenum}
\end{rrule}
\appendixproof{rrule:funnels:arc deletion set label}{
\begin{proof}[Correctness of \Cref{rrule:funnels:arc deletion set label}]
	Clearly, in a funnel the function~$\Label$ attributes every source a \Fork label and every sink a \Merge label.
	Since destroying sinks and sources is not possible,~\Cref{rrule:funnels:arc deletion set label} labels these vertices optimally.

	Let~$v$ be a vertex with~$\In{v} = 1$, let~$u \in \InN{v}$ be its only predecessor and assume~$L(u) = \Fork$.
	If we set~$L(v) \Set \Fork$, then \ArcDelApprox will not remove any arc when considering~$v$.
	If some outgoing arc~$(v,w)$ is removed, then necessarily~$L(w) = \Fork$.
	Hence, if we instead set~$L(v) \Set \Merge$ we also need to remove this arc, and potentially more.
	This implies that it is never worse to set~$L(v) \Set \Fork$ in this case.
	An analogous argument holds for the case where~$\Out{v} = 1$ and~$L(u) = \Merge$ for the only successor~$u$ of~$v$.

	Finally, let~$v$ be a vertex where~$\Out{v} = 1$,~$\In{v} > 1$, and~$\forall u \in \InN{v} : L(u) \neq \bot$.
	Since, by assumption, all outneighbors of~$v$ already have their labels set and satisfied, we only need to consider the label of~$v$ and of its only predecessor~$u$.
	If~$L(u) = \Fork$ in an optimal solution, then we know by the previous case that it is optimal to set~$L(v) \Set \Fork$.
	If~$L(u) = \Merge$ in an optimal solution, then we need to remove the arc~$(u,v)$ or some outgoing arc of~$v$.
	That is, we need to remove at least one arc of~$v$.
	By setting~$L(v) \Set \Fork$, we know that we need to remove exactly one arc of~$v$.
	Hence, doing so is optimal.
	An analogous argument also holds for the last case where we set~$L(v) \Set \Merge$.
\end{proof}}

Having labeled some vertices---whose labels will be as in an \optlab\ in some branch of the search tree---we simulate in \SatLabel the behavior of \ArcDelApprox and remove arcs from labeled vertices.
\begin{rrule}[\SatLabelname\ifappendix) ({\appref[\appsymb]{proof:rrule:funnels:arc deletion satisfy label}}\fi]
\label{rrule:funnels:arc deletion satisfy label}
Let~$v$ be some vertex where~$L(v) = \Fork$ and~$\In{v} > 1$.
If $\exists u \in \InN{v} : L(u) = \Fork$, then put the arcs $\{(x,v) \mid x \in \InN{v} \land x \neq u\}$ into the solution.
Otherwise, put $\{(x,v) \mid x \in \InN{v} \land L(x) = \Merge\}$ into the solution.

Let~$v$ be some vertex where~$L(v) = \Merge$ and~$\Out{v} > 1$.
If~$\exists u \in \OutN{v} : L(u) = \Merge$, then put the arcs~$\{(v,x) \mid x \in \OutN{v} \land x \neq u\}$ into the solution.
Otherwise, put~$\{(v,x) \mid x \in \OutN{v} \land L(x) = \Fork\}$ into the solution.
\end{rrule}
\appendixproof{rrule:funnels:arc deletion satisfy label}{
\begin{proof}[Correctness of \Cref{rrule:funnels:arc deletion set label}]
The arcs removed by \SatLabel would also be removed by \ArcDelApprox\ if all vertices had a label.
Hence, if the labels are correct, by \cref{prop:funnels:arc deletion approx=opt},
\SatLabel only removes arcs that are present in some optimal arc-deletion set.
\end{proof}}

To assign a label to each remaining vertex, we branch into assigning one of the two possible labels. Key to an efficient running time is the observation that there is always a vertex which, regardless of the label set, has some incident arc which then has to be in the solution. This observation is exploited in \cref{brule:funnels:arc deletion set label}.
\begin{brule}[Label Branch]
\label{brule:funnels:arc deletion set label}
If there is some vertex~$v$ such that~$\forall w \in \InN{v} : L(w) \neq \bot$ or~$\exists w \in \InN{v} : L(w) = \Fork$, then branch into two possibilities: Set~$L(v) \Set \Fork$; Set~$L(v) \Set \Merge$.

If there is some vertex~$v$ such that~$\forall w \in \OutN{v} : L(w) \neq \bot$ or~$\exists w \in \OutN{v} : L(w) = \Merge$, then branch into two possibilities: Set~$L(v) \Set \Fork$; Set~$L(v) \Set \Merge$.	
\end{brule}
The final \cref{brule:funnels:arc deletion satisfy label} tries all possibilities of satisfying a label of a vertex.
\begin{brule}[Arc Branch]
\label{brule:funnels:arc deletion satisfy label}
If there is a vertex~$v$ with~$L(v) = \Fork$ and $\In{v} > 1$, then branch into all possibilities of removing all but one incoming arc of~$v$.
If there is a vertex~$v$ with~$L(v) = \Merge$ and~$\Out{v} > 1$, then branch into all possibilities of removing all but one outgoing arc of~$v$.
\end{brule}
The correctness of \nameref{brule:funnels:arc deletion satisfy label} follows from \cref{prop:funnels:arc deletion approx=opt}.
To show the algorithm's correctness, it remains to show the following central lemma.%
\begin{lemma}
	\label{lemma:funnels:arc deletion alg label}
	Let~$D$ be a DAG.
	If \nameref{brule:funnels:arc deletion set label}, \nameref{brule:funnels:arc deletion satisfy label}, \SetLabel and \SatLabel are not applicable, then~$D$ is a funnel and all vertices have a label.
\end{lemma}
\begin{proof}
First, note that if the label of a vertex has been set, it will be satisfied by either applying \SatLabel or by branching with \nameref{brule:funnels:arc deletion satisfy label}.
Since satisfying all labels turns~$D$ into a funnel (\cref{thm:funnel:characterization}(\ref{chr:degree})), it is enough to show that all vertices have a label if \nameref{brule:funnels:arc deletion set label}, \SetLabel, and \SatLabel are not applicable. 

We first show that if there is some forbidden subgraph~$D' = (V',A') \subseteq D$, that is, $D'$ is isomorphic to some $D_i$ from \cref{thm:funnel:characterization}(\ref{chr:forbidden subgraph}), and if additionally \SetLabel and \SatLabel are not applicable, then \nameref{brule:funnels:arc deletion set label} is applicable.
Let~$D'$ be the forbidden subgraph in~$D$ with the smallest number of vertices.
Let~$v,u \in V'$ be two (not necessarily distinct) vertices in~$D'$ such that~$\In[D']{v} > 1$, $\Out[D']{u} > 1$.
Observe that all vertices between~$v$ and~$u$ 
in $D'$ (if any) have in- and outdegree one in~$D$, because $D'$ has the smallest number of vertices. 
We distinguish two cases.

\emph{Case 1:} $\forall w \in \InN[D]{v} : L(w) \neq \bot$.
Then either $\Out[D]{v} > 1$, meaning that we can apply \nameref{brule:funnels:arc deletion set label} (as required), or~$L(v) = \Merge$ due to \SetLabel.
Since all vertices between~$v$ and~$u$ have in- and outdegree one, we also know from the latter case that there is some arc~$(x,y)$ in the (uniquely defined) \Patht{v}{u} such that~$L(x) = \Merge$ and~$L(y) = \bot$.
Note that it cannot happen that~$L(y) = \Fork$ since \SatLabel is not applicable.
We also know that~$\Out[D]{y} > 1$ since \SetLabel is not applicable.
This implies \nameref{brule:funnels:arc deletion set label} is applicable on~$y$.

\emph{Case 2:} $\exists w \in \InN[D]{v} : L(w) = \bot$. This case is illustrated in \cref{fig:funnels:arc deletion labels lemma ex1}.
\begin{figure}[t]
\centering
\begin{tikzpicture}
\input{Pictures/arc-deletion-labels-lemma-ex1.tex}
\end{tikzpicture}
\caption{A DAG where \SatLabel and \SetLabel are not applicable.
The letter~F stands for a \Fork label and M stands for \Merge.
\nameref{brule:funnels:arc deletion set label} cannot be applied to~$v$ since~$u$ does not have a label, yet it can be applied to~$x \in \InC{w}$.}
\label{fig:funnels:arc deletion labels lemma ex1}
\end{figure}%
We show that we can find some vertex in~$\InC{w}$ to which we can apply \nameref{brule:funnels:arc deletion set label}.
Consider the longest \Patht{x}{w} that only contains vertices in~$\InC{w}$ which do not have a label.
Clearly, $\forall y \in \InN{x} : L(y) \neq \bot$ and~$\In{x} > 0$ since all sources have a label.
Thus, we can apply \nameref{brule:funnels:arc deletion set label} on~$x$.

Since only these two cases are possible, and in both we can apply \nameref{brule:funnels:arc deletion set label}, it follows, by contraposition, that~$D$ is a funnel and all vertices have a label if \nameref{brule:funnels:arc deletion set label}, \SetLabel, and \SatLabel are not applicable.
\end{proof}

By combining the previous data reduction and branching rules, we obtain a search-tree algorithm for \pADDF\ on DAGs:

\begin{theorem}
	\label{thm:funnels:arc deletion algorithm}
	\pADDF\ can be solved in time $\Bo(3^d \cdot (\Abs{V} + \Abs{A}))$, where $d$ is the arc-deletion distance to a funnel of a given DAG $D=(V,A)$.
\end{theorem}
\appendixproof{thm:funnels:arc deletion algorithm}{
  \begin{proof}
    The algorithm is as follows. On input of a DAG $D$, budget $d \in \mathbb{N}$, partial labeling~$L$, and partial solution~$A'$ (initially, $L$ does not label any vertex and $A' = \emptyset$), apply \SetLabel\ and \SatLabel\ until they do not apply anymore. If $|A'| > d$, then abort. Otherwise, apply \nameref{brule:funnels:arc deletion set label}, if possible. In each of the two resulting instances, apply \SatLabel\ until it does not apply anymore, and then apply \nameref{brule:funnels:arc deletion satisfy label}, if possible. Make a recursive call for each of the resulting instances. If no branching rule applies and $|A'| \leq d$, return $A'$ as a solution.

    By the correctness of the individual rules, the algorithm finds a solution if there is one (and otherwise does not return anything): From \cref{lemma:funnels:arc deletion alg label} we know that the algorithm turns the input into a funnel.
It remains to prove the running time bound.  
    With some simple bookkeeping and auxiliary tables, we can apply
    \SetLabel and \SatLabel\ 
    to all vertices
    of~$D$ 
    in total running time of~$\Bo(\Abs{V} + \Abs{A})$. 
    In the same running time we can find out whether \nameref{brule:funnels:arc deletion set label} and \nameref{brule:funnels:arc deletion satisfy label} is applicable. Hence, we need at most $\Bo(\Abs{V} + \Abs{A})$ running time per recursive call.
    

    It remains to bound the size of the search tree, that is, the outtree~$\cal T$ whose vertices are the calls of the algorithm and whose edges represent recursive calling relation. To ease the analysis, we instead bound the modified tree~$\cal T'$ in which we replace the outneighbors of a vertex corresponding to the recursive calls resulting from \nameref{brule:funnels:arc deletion satisfy label} by a binary tree as follows. If \nameref{brule:funnels:arc deletion satisfy label} branches into all possibilities of putting into the solution a subset of size~$d' - 1$ from an arc set~$B$ of size~$d'$, we instead recursively choose two arcs and introduce two recursive calls in which one of the two arcs is put into the solution until $B$ has size~1. Clearly, the size of $\cal T$ is upper bounded by the size of~$\cal T'$.

    We claim that $\cal T'$ has maximum outdegree three. Consider the instances resulting from \nameref{brule:funnels:arc deletion set label}. Without loss of generality assume that the first portion of \nameref{brule:funnels:arc deletion set label} was applied. The proof for the second portion is analogous. If $L(v)$ was set to $\Fork$, then after \SatLabel\ has been applied exhaustively, \nameref{brule:funnels:arc deletion set label} is not applicable. Otherwise, if $L(v)$ was set to \Merge, the modified \nameref{brule:funnels:arc deletion set label} with two branches is applied. Hence, indeed, there are at most three recursive calls.

    To bound the size of $\cal T'$, consider a path~$P$ from the root to a leaf. Whenever \nameref{brule:funnels:arc deletion set label} is applied, in each of the following recursive calls, at least one arc is put into the solution: Without loss of generality, assume that the first portion of \nameref{brule:funnels:arc deletion set label} was applied. The proof for the second portion is analogous. If $L(v)$ was set to \Fork, then \SatLabel will put at least one arc into the solution (note that, since \SetLabel is not applicable, $v$ has indegree at least two). Otherwise, if $L(v)$ was set to \Merge, then, since \SetLabel is not applicable, \nameref{brule:funnels:arc deletion satisfy label} is applicable and will put at least one outarc of $v$ into the solution. Hence, $P$ has length at most $d$, since no further recursive calls are made if $|A'| > d$. Combining this with the fact that $\cal T'$ has outdegree at most three, it follows that $\cal T'$ has size $O(3^d)$. 

	Hence, the running time of the search-tree algorithm is~$\Bo(3^d \cdot (\Abs{V} + \Abs{A}))$
        .
      \end{proof}
      }

\looseness=-1 To improve the running time of the search-tree algorithm in practice, we compute a lower bound of the arc-deletion distance to a funnel of the input and we stop expanding a branch of the search tree when the lower bound exceeds the available budget. 
A simple method for computing a lower bound is to find arc-disjoint forbidden subgraphs.
Clearly, the sum of the arc-deletion distances to a funnel of the subgraphs found is not larger than the distance of the input DAG.
To find such subgraphs, we first look for vertices with both in- and outdegree greater than one, which are not allowed in funnels.
Then we search for paths~$v_1,v_2, \dots, v_k$ such that~$\In{v_1} > 1$ and~$\Out{v_k} > 1$.
With some bookkeeping we can find a maximal set of arc-disjoint forbidden subgraphs in linear time.




\section{Empirical Evaluation of the Developed Algorithms}
\label{sec:experiments}
In this section, we empirically evaluate the approximation algorithm and the fixed-parameter algorithm for  \pADDF\ described in \cref{sec:algs}. We used artificial data sets and data based on publicly available real-world graphs. 
Our experiments show that both our algorithms are efficient in practice.

We implemented the algorithms in Haskell 2010. 
All experiments were run on an Intel$^\circledR$ Xeon$^\circledR$ E5-1620 \SI{3.6}{\giga\hertz} processor with \SI{64}{\giga\byte} of RAM.
The operating system was GNU/Linux, with kernel version 4.4.0-67.
For compiling the code, we used GHC version 7.10.3.
The code is released as free software \cite{parfunn}.

\paragraph{Experiments on Synthetic Funnel-like DAGs.}
\label{ssec:experiment:distance:funnel-like}
We generated random funnel-like DAGs through the following steps.
\begin{inparaenum}[(1)]
\item Choose the number of vertices, arc density~$p \in [0,1]$, and some~$s \in \Nat$.
\item Fix a topological ordering of the vertices.
\item Uniformly at random assign a label \Fork or \Merge to each vertex. 
\item Create an out-forest with \Fork vertices, and an in-forest with \Merge vertices.
\item Add random arcs from \Fork to \Merge vertices until a density of~$p$ (relative to the maximum number of arcs allowed by the labeling) is achieved.
\item Add~$s$ random arcs	which respect the topological ordering.
\end{inparaenum}
Steps (1) through~(5) result in a funnel which we call \emph{planted} funnel below.

For a fixed labeling, the algorithm above generates funnels uniformly at random from the input parameters.
The labeling, however, is drawn uniformly at random from all~$2^{\Abs{V}}$ possible labelings, without considering how many different funnels exist with a given labeling.
Hence, funnels with fewer arcs have a larger chance of being generated than funnels with many arcs (when compared to the chances in a uniform distribution).
We consider this bias to be harmless for the experiments since, for the exact algorithm, the number of arcs is not decisive for the running time, and for the approximation algorithm the number of arcs should not have a big impact on the solution quality. 
%

For~$n \in \{250, 300, 500, 1000\}$,~$p \in \{0.15, 0.5, 0.85\}$ and~$s \in \{125, 150, 175\}$ we generated 30 funnels with~$n$ vertices and density~$p$, and then added~$s$ random arcs as described above.
This gives us a total of~$1080$ DAGs.

Our fixed-parameter algorithm was able to compute the arc-deletion distance to funnel of~$1059$ instances ($98\%$) within 10 minutes.
The approximation algorithm finished on average in less than \SI{72}{\milli\second}.
\begin{figure}[t]
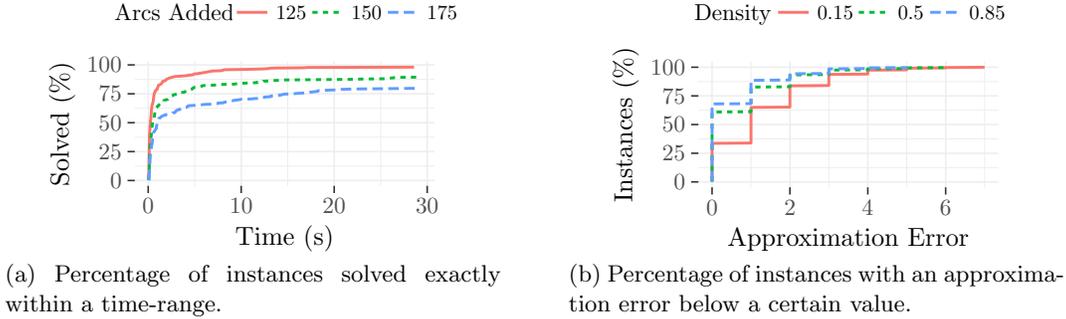

\centering
\SubfigureL
{\input{Plots/robustness-01.tex}
\label{fig:distance:funnel-like:time vs solved}}
{Percentage of instances solved exactly within a time-range.
}
{\Half}
\qquad
\SubfigureL
{\input{Plots/robustness-06.tex}
\label{fig:distance:funnel-like:approx error}}
{Percentage of instances with an approximation error below a certain value.} 
{\Half}
\caption{Running time and approximation error.}
\end{figure}
A cumulative curve with the percentage of instances solved within a certain time range is depicted in \cref{fig:distance:funnel-like:time vs solved}.
Most instances were solved fairly quickly: Within 15 seconds 932 ($86\%$) instances were solved optimally.
We can also observe that there were essentially two types of instances: Easy ones which were solved within few seconds, and harder ones which often were not solved within 10 minutes.
That is, if we limit the running time to five seconds, then we can solve 856 (79\%) instances, and if we increase it to sixty seconds, we can solve only 141 additional instances.

\cref{fig:distance:funnel-like:approx error} shows the relation between the error of the approximation algorithm with the density of the planted funnel.
The approximation algorithm found an optimal solution in 574 (54\%) instances, and in 260 (25\%) it removed only one more arc than necessary.
As the arc-deletion distance to a funnel of most instances was greater than 100, this means that the approximation ratio is very close to one.
Since the DAGs used here are already close to funnels, most decisions of the approximation algorithm are correct.
Intuitively, having correct local information helps the approximation make a globally optimal decision, and so it is unsurprising that the approximation factor in funnel-like DAGs is much better than the theoretical bound.
This is supported also by the fact that the approximation performed worse on sparse planted funnels than on dense ones
, since the proportion of ``wrong'' information regarding the arcs is larger on sparse funnels (when adding the same number of random arcs).

\paragraph{Experiments on DAGs Based on Real-World Data Sets.}
\label{ssec:experiment:distance:real world}
We obtained ten digraphs from the Konect database~\cite{KonectDataset}, containing food-chains, interactions between animals, and source-code dependencies.
We also downloaded the dependency network of all packages in Arch Linux.\footnote{Listed at \url{https://www.archlinux.org/packages/} and obtained using \texttt{pacman}.}
Since most of the gathered digraphs contain cycles, we performed a pre-processing step turning them into DAGs: we merged cycles into a single vertex, and then removed self-loops.
For each of the eleven DAGs we computed a lower bound and an approximation of its arc-deletion distance to funnel.
We also attempted to compute the real distance, stopping the algorithm if no solution was found within four hours.

The dataset was divided into six small DAGs ($\leq 156$ vertices and~$\leq 1197$ arcs) and five larger ones ($\geq 5730$ vertices and~$\geq 26218$ arcs). 
In the small ones, our fixed-parameter algorithm solved \pADDF\ within one second, and our approximation algorithm found the correct distance in $\leq$ \SI{2}{\milli\second}.
In two of the six small DAGs the distance was 60 and 129, which means that the exact algorithm is in practice much faster than what the worst-case upper bound predicts.

\looseness=-1 On the larger DAGs the fixed-parameter algorithm could not solve \pADDF\ within four hours.
By computing a lower bound for the distance, we managed to give an upper bound for the approximation factor, which was at most~$1.16$.
This means that the approximation algorithm is practical since it is fast ($\leq\SI{228}{\milli\second}$ on average) and yields a near-optimal solution.
Relative to the number of arcs, the arc-deletion distance to a funnel parameter was small (9\% on average).



\section{Conclusion}
\label{sec:conclusion}
\looseness=-1 We believe that our results add to the relatively small list of
fixed-parameter tractability results for directed graphs and introduce
a novel interesting structural parameter for directed (acyclic)
graphs. In particular, our approximation and fixed-parameter
algorithms could help to establish the arc-deletion distance to a funnel as a useful
``distance-to-triviality measure''~\cite{Cai03,GHN04,Nie10} for
designing fixed-parameter algorithms for NP-hard problems on DAGs.
We leave open whether computing the arc-deletion distance to funnel of a DAG is APX-hard. 
Finally, funnels might provide a basis for defining some useful 
digraph width or depth measures~\cite{GHKLOR14,GHK0ORS16,Mil17}.

\ifarxiv
	\bibliographystyle{abbrvnat} 
\else
	\bibliographystyle{splncs03}
\fi
\bibliography{literature}

\ifappendix%

\newpage

{\Large \textbf{Appendix}}

\appendix



	\appendixProofText
\fi

\end{document}